\documentclass{llncs}
\usepackage[utf8x]{inputenc}
\usepackage{amsmath}
\usepackage{amssymb}
\usepackage{amsfonts}
\usepackage{algorithmic}
\usepackage{algorithm}
\usepackage{hyperref}

\usepackage{enumerate}

\newtheorem{thm}{Theorem}
\newtheorem{lmm}{Lemma}
\newcommand{\Z}{\mathbb{Z}}

\begin{document}
\title{Benaloh's Dense Probabilistic Encryption Revisited\thanks{This work was supported by ANR SeSur SCALP, SFINCS, AVOTE.}}
\author{
Laurent Fousse\inst{1}
\and
Pascal Lafourcade\inst{2}
\and
Mohamed Alnuaimi\inst{3}
}

\institute{
Université Grenoble 1\\
CNRS \\
Laboratoire Jean Kuntzmann\\
France\\
\email{Laurent.Fousse@imag.fr}
\and
Université Grenoble 1\\
CNRS \\
Verimag \\
France\\
\email{Pascal.Lafourcade@imag.fr}
\and
ENSIMAG\\
France\\
\email{alnuaimi.mohd@gmail.com}
}

\date{}
\maketitle

\pagestyle{plain}

\begin{abstract}
  In 1994, Josh Benaloh proposed a probabilistic homomorphic
  encryption scheme,
enhancing the poor expansion
  factor provided by Goldwasser and Micali's scheme.
  Since then, numerous papers have taken
  advantage of Benaloh's homomorphic encryption function, including
  voting schemes, computing multi-party trust privately,
  non-interactive verifiable secret sharing, online poker...  In this
  paper we show that the original description of the scheme is
  incorrect, possibly resulting in ambiguous decryption of
  ciphertexts. We give a corrected description of the scheme and provide
  a complete proof of correctness. We also compute the probability of
  failure of the original scheme. Finally we analyze several
  applications using Benaloh's encryption scheme.
  We show in each case the impact of a bad choice in the key generation
  phase of Benaloh's scheme. For instance in the application of e-voting
  protocol, it can inverse the result of an election, which is a non
  negligible consequence.
\end{abstract}

\section{Introduction}

In the literature there are several homomorphic encryption schemes as
for instance schemes proposed by
Goldwasser-Micali~\cite{GoldwasserM82},
ElGamal~\cite{elgamal_proceedings_1985}, Benaloh~\cite{Benaloh94},
Naccache and Stern~\cite{NaccacheS98}, Okamoto and
Uchiyama~\cite{OU98}, Paillier~\cite{Paillier:1999:PKC} and its
generalization proposed by Damg\aa rd and
Jurik~\cite{damgaard_generalisationsimpli._2001}, Sander, Young and
Yung~\cite{SanderYY99}, Gaborit and Aguillar~\cite{MelchorGH10}. In
this paper we focus our attention on Benaloh's encryption
scheme. In~\cite{CFFG} a survey of existing homomorphic encryption
schemes is proposed for the non specialist.
In~\cite{akinwande_advances_2009} the author also proposes a
description and a complexity analysis of different existing
homomorphic encryption schemes. In~\cite{rappe_homomorphic_2004}, the
author considers homomorphic cryptosystems and their applications. In
all these papers authors mention existing homomorphic encryption
schemes and give descriptions of such schemes including Benaloh's
scheme.  Homomorphic encryption schemes have several 
applications. We only cite applications that are using Benaloh's
scheme as for example voting schemes~\cite{BT94,RuizV05,ben87},
computing multi-party trust privately~\cite{cramer01,PIR05,Trust10},
non-interactive verifiable secret sharing~\cite{ben87}, online
poker~\cite{Golle:2005}... 

Despite all these papers on applications and implementations realized
by all these specialists, we were surprised to discover that the
condition in the key generation of Benaloh's scheme can in some cases
lead to an ambiguous encryption. How is it possible that after all
these papers, results, protocols, even implementations and more than
fifteen years nobody noticed it?  In order to answer this question we
will explicitly express the failure probability of the original scheme
in Section~\ref{sec:proba}. How did we discover this problem?  We
wanted to perform a time comparison of the efficiency of some
well-known homomorphic encryptions.  We proposed a methodology for
testing their performance on large randomly generated data. So we
started to code some of the encryption and decryption functions of
homomorphic primitives.  Benaloh's was one of the first one that we
have tried. We were surprised to see that on some randomly generated
instances of Benaloh's parameters our decryption function was not able
to recover the correct plaintext. After verifying our code several
times according to the conditions given in the original paper we were
not able to understand why our code did not give the right
plaintext. So we investigated more and were able to generate several
counter-examples (one example of problematic parameters is given in
Section~\ref{sec:counterex}) and more interestingly we clearly
understood why and where the scheme failed. Indeed the bug is due to a
very small detail, hence we proposed a revisited version of Benaloh's
dense probabilistic encryption.

\paragraph{Contributions: } The first contribution is that the original
scheme proposed by Benaloh in~\cite{Benaloh94} does not give a unique
decryption for all ciphertexts. We exhibit a simple example in the
rest of the paper and characterize when this can happen and how to
produce such counter-examples.  Indeed the problem comes from the
condition in the generation of the public key. The condition is not
strong enough and allows to generate such keys that can for some
plaintexts generate ambiguous ciphertexts.

Hence our second contribution is a new condition for the key generation
which avoids such problem. We not only propose a new correct condition
but also give an equivalent practical condition that can be used for
implementations. We also compute the probability of failure of the
original scheme, in order to justify why nobody discovered the problem
before us.

Finally we describe some applications using explicitly Benaloh's
scheme. In each case we briefly explain how the application works on a
simple example. With these examples we clearly show that if our new
condition is not used then the wrong key generation can have important
consequences. In the case of the e-voting protocol it can change the
result of an election; for private multi-party trust computation it
can really impact the trust that somebody can have in someone.

\paragraph{Outline:} In Section~\ref{sec:original} we recall the
original Benaloh scheme.  In Section~\ref{sec:counterex} we give a
small example of parameters following the initial description and
where we have ambiguous decryption. Then in
Section~\ref{sec:correction} we give a corrected description of the
scheme, with a proof of correctness. The probability of choosing
incorrect parameters in the initial scheme is discussed in
Section~\ref{sec:proba}.  In Section~\ref{sec:rel} we discuss some
schemes related to Benaloh's scheme.  Finally before concluding in the
last section, we demonstrate using some applications that the problem
we discover can have serious consequences.

\section{Original Description of Benaloh's Scheme} \label{sec:original}

Benaloh's ``Dense Probabilistic Encryption''~\cite{Benaloh94}
describes an homomorphic encryption scheme with a significant
improvement in terms of expansion factor compared to
Goldwasser-Micali~\cite{GoldwasserM82}. For the same security
parameter (the size of the RSA modulus $n$), the ciphertext is in both
cases an integer $\bmod~n$, but the cleartext in Benaloh's scheme is
an integer $\bmod~r$ for some parameter $r$ depending on the key,
whereas the cleartext in Goldwasser-Micali is only a bit. When
computing the expansion factor for random keys, we found that it is
most of the times close to $1/2$ while it is $\lceil \log_2(n)\rceil$
for Goldwasser-Micali. We now recall the three steps of the original
scheme given in Benaloh's paper~\cite{Benaloh94}.

\paragraph{Key Generation}

The public and private key are generated as follows:
\begin{itemize}
\item Choose a block size $r$ and two large primes $p$ and $q$ such
  that :
\begin{itemize}
\item $r$ divides $(p-1)$.
\item $r$ and $(p-1)/r$ are relatively prime.
\item $r$ and $q-1$ are relatively prime.
\item $n = pq$.
\end{itemize}
\item Select $y \in (\mathbb{Z}_n)^* = \{ x \in \mathbb{Z}_n :
  \gcd(x,n)=1\}$ such that 
\begin{equation}
y^{\varphi/r} \neq 1 \bmod n
\label{buggy_y_condition}
\end{equation}
where  $\varphi$ denotes $ (p-1)(q-1)$.
\end{itemize}
The public key is $(y, r, n)$, and the private key is the two primes
$p$, $q$.

\paragraph{Encryption}

If $m$ is an element in $\mathbb{Z}_r$ and $u$ a random number in
$(\mathbb{Z}_n)^*$ then we compute the randomized encryption of $m$
using the following  formula:
\[ E_r(m) = \{y^m u^r \bmod n : u \in (\mathbb{Z}_n)^*\}.\]

\paragraph{Decryption} We first notice that for any $m$, $u$ we have: 
\[    (y^m u^r)^{(p-1)(q-1)/r} \equiv y^{m(p-1)(q-1)/r} u^{(p-1)(q-1)}
\equiv y^{m(p-1)(q-1)/r} \mod n.\]

Since $m < r$ and $y^{(p-1)(q-1)/r} \not \equiv 1 \mod n$, Benaloh
concludes that $m = 0$ if and only if $(y^m u^r)^{(p-1)(q-1)/r} \equiv
1 \mod n$.  So if $z = y^m u^r \mod n$ is an encryption of $m$, given
the secret key $p$, $q$ we can determine whether $m=0$. If $r$ is
small, we can decrypt $z$ by doing an exhaustive search of the
smallest non-negative integer $m$ such that $(y^{-m}z \mod n) \in
E_r(0)$.  By precomputing values and using the baby-step giant-step
algorithm it is possible to perform the decryption in time
$O(\sqrt{r})$. Finally if $r$ is smooth we can use classical
index-calculus techniques. More details about these optimization of
decryption are discussed in the original paper~\cite{Benaloh94}.

We remark that there is a balance to find between three parameters
in this cryptosystem:
\begin{itemize}
\item
  ease of decryption, which requires that $r$ is a product of small
  prime powers,
\item
  a small expansion factor, defined as the ratio between the size of
  the ciphertexts and the size of the cleartexts. Because $p$ and $q$
  have the same size and $r\mid p-1$, this expansion factor is at
  least $2$,
\item
  robustness of the private key, meaning that $n$ should be hard to
  factor. In the context of the P-1 factorization
  method~\cite{Pollard74}, a big smooth factor of $p-1$ is a definite
  weakness.  
\end{itemize}
We notice that the cryptosystem of Naccache-Stern~\cite{NaccacheS98},
similar to Benaloh's scheme, addresses this issue and by consequence do
not produce ambiguous encryption.

\section{A Small Counter-Example}
\label{sec:counterex}

We start by picking a secret key $n = pq = 241\times{}179 = 43139$, for which
we can pick $r=15$. Algorithm~\ref{alg:compr} may be used to
compute the maximal suitable value of the $r$ parameter if you start
by picking $p$ and $q$ at random, but a smaller and smoother value may
be used instead for an easier decryption.

\begin{algorithm}
\begin{algorithmic}
\STATE $r \gets p-1$
\WHILE{$\gcd(q-1, r) \neq 1$}
\STATE $r \gets r/\gcd(r, q-1)$
\ENDWHILE
\end{algorithmic}
\caption{Compute $r$ from $p$ and $q$.}\label{alg:compr}
\end{algorithm}

We verify that $r=15$ divides $p-1=240=16\times{}15$, $r$ and
$(p-1)/r=16$ are relatively prime, $r=15=3\times{}5$ and
$q-1=178=2\times{}89$ are coprime.  Assume we pick $y=27$, then
$\gcd(y,n)=1$ and $y^{(p-1)(q-1)/r} = 40097 \neq 1 \bmod n$ so
according to Benaloh's key generation procedure all the original
conditions are satisfied.

By definition, $z_1 = y^112^{r} = 24187$ is a valid encryption of
$m_1=1$, while $z_2 = y^64^{r} = 24187 = z_1$ is also a valid
encryption of $m_2=6$. In fact we can verify that with this choice of
$y$, the true cleartext space is now $\Z_5$ instead of $\Z_{15}$
(hence the ambiguity in decryption): first notice that in $\Z_p$,
$y^5=8=41^{15}$. This means that a valid encryption of $5$ is
also a valid encryption of $0$. For any message $m$, the set of
encryptions of $m$ is the same as the set of encryptions of $m+5$,
hence the collapse in message space size. The fact that the message
space size does not collapse further can be checked by brute force
with this small set of parameters.

For this specific choice of $p$ and $q$, there are
$\frac{r-1}{r}\varphi(n)=39872$ possible values of $y$ according to
the original paper, but $17088$ of them would lead to an ambiguity in
decryption (that's a ratio of $3/7$), sometimes decreasing the
cleartext space to $\Z_3$ or $\Z_5$.  Details are provided in
Section~\ref{sec:proba}.

\section{Corrected Version of Benaloh's Scheme}
\label{sec:correction}

Let $g$ be a generator of the group $\Z_p^*$, and since $y$ is coprime
with $n$, write $y = g^\alpha \bmod p$. We will now state in
Theorem~\ref{benaloh_fixed} our main contribution:
\begin{thm}
\label{benaloh_fixed}
The following properties are equivalent:
\begin{enumerate}[a)]
  \item $\alpha$ and $r$ are coprime;
    \label{alpha_r_coprime}
  \item decryption works unambiguously;
    \label{dec_works}
  \item For all prime factors $s$ of $r$, we have $y^{(\varphi/s)} \neq 1 \bmod n$.
    \label{y_good_order}
\end{enumerate}
\end{thm}
Of course property (\ref{dec_works}) is what we expect of the scheme, while
(\ref{alpha_r_coprime}) is useful to analyze the proportion of invalid
$y$'s and (\ref{y_good_order}) is more efficient to verify in practice
than (\ref{alpha_r_coprime}), especially considering that in order to decrypt
efficiently the factorization of $r$ is assumed to be known.
\begin{proof}
We start by showing $(\ref{alpha_r_coprime}) \Rightarrow
(\ref{dec_works})$. Assume two messages $m_1$ and $m_2$ are encrypted
to the same element using nonce $u_1$ and $u_2$:
\[ y^{m_1}u_1^r = y^{m_2}u_2^r \bmod n. \]
Reducing $\bmod~p$ we get:
\[ g^{\alpha(m_1-m_2)} = (u_2/u_1)^r \bmod p \]
and using the fact that $g$ is a generator of $(\Z/p\Z)^*$, there exists some $\beta$ such that
\[ g^{\alpha(m_1-m_2)} = g^{\beta{}r} \bmod p \]
which in turns implies
\[ \alpha(m_1-m_2) = \beta{}r \bmod p-1.\]
By construction of $r$, we can further reduce $\bmod~r$ and get
\[ \alpha(m_1-m_2) = 0 \bmod r \]
and since $r$ and $\alpha$ are coprime, we can deduce $m_1=m_2 \bmod r$, which means that decryption
works unambiguously since the cleartexts are defined $\bmod~r$.

\vskip 5mm 
We now prove that $(\ref{alpha_r_coprime}) \Rightarrow
(\ref{y_good_order})$. Assume that there exists some prime factor $s$
of $r$ such that
 \[ y^{(\varphi/s)} = 1 \bmod n. \]
As above, by reducing $\bmod~p$ and using the generator $g$ of
$(\Z/p\Z)^*$ we get
\begin{equation}
\alpha\frac{\varphi}{s} = 0 \bmod p-1.
\label{eq:toto}
\end{equation}
Let $k = v_s(r)$ the $s$-valuation of $p-1$ and write $\alpha = \nu{}s +
\mu$ the Euclidean division of $\alpha$ by $s$. By construction we have
$v_s(\varphi) = k$. When reducing (\ref{eq:toto}) $\bmod~s^k$ we can remove all
factors of $\varphi$ that are coprime with $s$, so we get
\begin{eqnarray*}
\alpha{}s^{k-1} & = & 0 \bmod s^k \\
\mu{}s^{k-1} & = & 0 \bmod s^k \\
\mu{} & = & 0 \bmod s \\
\mu{} & = & 0
\end{eqnarray*}
and $\alpha$ and $r$ are not coprime.

\vskip 5mm

We now prove $(\ref{y_good_order}) \Rightarrow
(\ref{alpha_r_coprime})$. Assume $\alpha$ and $r$ are not coprime and
denote by $s$ some common prime factor.
Then
\begin{eqnarray*}
y^{(\varphi/s)} & = & g^{\alpha\varphi/s} \bmod p \\
& = & g^{(\alpha/s)\varphi} \bmod p = 1 \bmod p.
\end{eqnarray*}
And by construction of $r$, $s \nmid q-1$ so $y^{(\varphi/s)} = 1 \bmod q$.

\vskip 5mm
We now prove $(\ref{dec_works}) \Rightarrow
(\ref{alpha_r_coprime})$. Assume two different cleartexts $m_1 \neq
m_2 \bmod r$ are encrypted to the same ciphertext using nonces $u_1$
and $u_2$:
\[ y^{m_1}u_1^r = y^{m_2}u_2^r \bmod n. \]
As before, we focus on operations $\bmod~p$ and we get
\[ \alpha(m_1 - m_2) = 0 \bmod r.\]
If $\alpha$ were invertible $\bmod~r$, we would get an absurdity.
\end{proof}

Notice than in the  example of Section~\ref{sec:counterex} we have $y^{(p-1)(q-1)/3}=1 \bmod n$
so condition (\ref{y_good_order}) is not satisfied. We claimed that
the real ciphertext space is now $\Z_5$, and we give a precise
analysis of the cleartext space reduction at the end of
Section~\ref{sec:proba}.

\section{Probability of Failure of Benaloh's Scheme}
\label{sec:proba}

We now estimate the probability of failure in the scheme as originally
described. For this we need to count the numbers $y$ that satisfy
condition (\ref{buggy_y_condition}) and not property (\ref{y_good_order}) of
Theorem \ref{benaloh_fixed}. We call these values of $y$ ``faulty''.

\begin{lmm}
Condition (\ref{buggy_y_condition}) is equivalent to the statement: $r \nmid \alpha$.
\end{lmm}
\begin{proof}
Assume that $r$ divides $\alpha$: $\alpha = r\alpha'$. So
\begin{eqnarray*}
  y^{\varphi/r} & = & g^{\alpha\varphi/r} \bmod n\\
  & = & (g^{\alpha'})^\varphi \bmod n \\
  & = & 1 \bmod n.
\end{eqnarray*}
Conversely, if $y^{\varphi/r}=1 \bmod n$, then
\begin{eqnarray*}
  g^{\alpha\varphi/r} & = & 1 \bmod n\\
  & = & 1 \bmod p\\
  \alpha\frac{\varphi}{r} & = & 0 \bmod{p-1}.
\end{eqnarray*}
Since $r$ divides $p-1$ and is coprime with $\frac{\varphi}{r}$ (by
definition), we have $r\mid \alpha$.
\end{proof}

Since picking $y\in (\Z_p)^*$ at random is the same when seen
$\bmod~p$ as picking $\alpha \in \{0, \ldots, p-2 \}$ at random, we
can therefore conclude that the proportion $\rho$ of faulty $y$'s is
exactly the proportion of non-invertible numbers $\bmod~r$ among the non-zero
$\bmod~r$. So 
$\rho = 1-\frac{\varphi(r)}{r-1}$. We notice
that this proportion depends on $r$ only, and it is non-zero when $r$ is not
a prime. Since decryption in Benaloh's scheme is essentially solving a
discrete logarithm in a subgroup of $\Z_p$ of order $r$, the original schemes
recommends to use $r$ as a product of small primes' powers, which tends to increase
$\rho$. In fact, denoting by $(p_i)$ the prime divisors of $r$
we have:
\[ \rho = 1-\frac{r}{r-1} \prod_i \frac{p_i-1}{p_i}  \approx 1-\prod_i \frac{p_i-1}{p_i} \]
which shows that the situation where decryption is easy also increases the proportion
of invalid $y$ when using the initial description of the encryption scheme.

As a practical example, assume we pick two $512$ bits primes $p$ and
$q$ as
\begin{eqnarray*}
p & = & 2\times (3\times{}5\times{}7\times{}11\times{}13) \times p' + 1\\
p' & = & 4464804505475390309548459872862419622870251688508955\verb+\+ \\
   &   & 5037374496982090456310601222033972275385171173585381\verb+\+ \\
   &   & 3914691524677018107022404660225439441679953592\\
q  & = & 1005585594745694782468051874865438459560952436544429\verb+\+ \\
   &   & 5033292671082791323022555160232601405723625177570767\verb+\+ \\
   &   & 523893639864538140315412108959927459825236754568279.\\
\end{eqnarray*}
Then
\begin{eqnarray*}
\gcd(q-1, p-1) & = & 2 \\
r & = & (3\times{}5\times{}7\times{}11\times{}13) \times p'\\
\rho & = &
1-\frac{r}{r-1}\times\frac{2}{3}\times\frac{4}{5}\times\frac{6}{7}\times\frac{10}{11}\times\frac{12}{13}\times\frac{p'}{p'-1}\\
\rho     & > & 61\%.
\end{eqnarray*}
This example was constructed quite easily: first we take $p'$ of
suitable size, and increase its value until $p$ is prime. Then we
generate random primes $q$ of suitable size until the condition
$\gcd(p-1,q-1)=2$ is verified; it took less than a second on a current
laptop using Sage~\cite{sage}.

Putting it all together, we can also characterize the faulty values of
$y$, together with the actual value $r'$ of the cleartext space size
(compared to the expected value $r$):
\begin{lmm}
Let $u = \gcd(\alpha, r)$. Then $r' = \frac{r}{u}$.
Moreover if $r' \neq r$, this faulty value of $y$ goes undetected by
the initial condition as long as $u\neq r$.
\end{lmm}
The proof of the first implication in Theorem~\ref{benaloh_fixed} is
easily extended to a proof of the first point of this lemma, while the
second point is a mere rephrasing of the previous lemma.

This result can be used to craft counter-examples as we did in
Section~\ref{sec:counterex}: for a valid value $y$ of the parameter
and $u$ a proper divisor of $r$, the value $y'=y^u\bmod n$ is an
undetected faulty value with actual cleartext space size $r'=r/u$. It
can also be used to determine precisely, for every proper divisor $r'$
of $r$ the probability of picking an undetected faulty parameter $y$
of actual cleartext space size $r'$. Such an extensive study was not
deemed necessary in the examples to follow in Section~\ref{sec:app}.

\section{Related Schemes}
\label{sec:rel}

We briefly discuss in this section some schemes related to that
of~\cite{Benaloh94}.

In~\cite{BT94}, the authors describe a cryptosystem which closely
resembles that of~\cite{Benaloh94}, but the conditions given on $r$ are
less strict. Let us recall briefly the parameters of the cryptosystem
as described in~\cite{BT94}:
\begin{itemize}
\item
  $r \mid p-1$ but $r^2 \nmid p-1$.
\item
  $r \nmid q-1$.
\item
  $y$ is coprime with $n$ and $y^{(p-1)(q-1)/r} \neq 1 \bmod n$.
\end{itemize}
It is clear that $r^2 \nmid p-1$ is weaker than $\gcd((p-1)/r, r) =
1$, and that $r\nmid q-1$ is weaker than $\gcd(q-1, r)=1$.
Therefore any set of parameters satisfying~\cite{Benaloh94} are also
valid parameters as defined in~\cite{BT94}.

Unfortunately the condition imposed on $y$ is the same
and still insufficient, and finding counter-examples is again a matter of
picking $\alpha$ not coprime with $r$. Our theorem still stands for
this cryptosystem if you replace condition~(\ref{y_good_order}) by the
following condition:
\begin{equation}
\mbox{For all prime factors $s$ of $r$, we have $y^{(p-1)/s} \neq 1 \bmod p$.}
\label{y_good_order2}
\end{equation}

Going back in time, the scheme of Goldwasser and
Micali~\cite{GoldwasserM82} can be seen as a precursor of~\cite{BT94}
with a fixed choice of $r=2$. The choice of $y$
in~\cite{GoldwasserM82} as a quadratic non-residue $\bmod n$ is
clearly an equivalent formulation of condition~(\ref{y_good_order2}).

Before~\cite{Benaloh94} and~\cite{BT94}, the scheme was defined by
Benaloh in~\cite{Benaloh87}, with the parameter $r$ being a prime. In
this case our condition~(\ref{y_good_order}) is the same as the one
proposed by Benaloh, and the scheme in this thesis is indeed
correct. The main difference between the different versions proposed
afterwards and this one is that it is not required for $r$ to be
prime, which leads in some cases to ambiguous ciphers. This remark
clearly shows that all details are important in cryptography and that
the problem we discover is subtle because even Benaloh himself did not
notice it.

Finally the scheme proposed by Naccache and Stern~\cite{NaccacheS98}
is quite close to the one proposed in~\cite{Benaloh87} but with a
parameterization of $p$ and $q$. It makes decryption correct,
efficient, and leaves the expansion factor as an explicit function of
the desired security level with respect to the $P-1$ method of
factoring~\cite{Pollard74} (the expansion is essentially the added
size of the big cofactors of $p-1$ and $q-1$). We note in passing that
a modulus size of 768 bits was considered secure at the time, a fact
disproved twelve years later~\cite{Kleinjung10factorizationof} only!

\section{Applications}
\label{sec:app}

In this last section, we present some applications which explicitly
use Benaloh's encryption scheme. We analyze in each situation what are
the consequence on the application of using a bad parameter produced
during the key generation.

\subsection{Receipt-free Elections}

In~\cite{BT94} the authors propose an application of homomorphic
encryption for designing new receipt-free secret-ballot
elections. They describe two protocols which use an homomorphic
encryption verifying a list of properties. They also give in Appendix
of the paper a precise description of a encryption scheme which
satisfies their properties. Its relation with~\cite{Benaloh94}
is given in section~\ref{sec:rel}.

 The new voting protocol uses the fact that the encryption is
 homomorphic and probabilistic. If we have two candidates Nicolas and
 Ségolène then the system associates for instance the ballot $0$ for
 Nicolas and the ballot $1$ for Ségolène. The main idea is that the
 server collects the $m$ authenticated encrypted ballots $\{v_i\}_k$
 corresponding to the choices $v_i$ of the $m$ voters. Hence the
 server performs the multiplication of all these votes and decrypts
 the product once to obtain the result. The number obtained
 corresponds to the number of votes for Ségolène $n_S$ and the
 difference $m-n_S$ gives the number of votes for Nicolas.

We construct a basic application of the first protocol proposed
in~\cite{BT94} and based on the example described in
Section~\ref{sec:counterex}. In this example we consider only $20$
voters. If the encryption is correctly done then the final result is
$\{14\}_k$. It means that after decryption Ségolène has $14$ votes and
Nicolas has $6$ votes. But if as we explain in
Section~\ref{sec:counterex} instead of computing the result $14 \bmod
15$ we are taking the result modulo $5$, then we obtain a result of
$14 \bmod 5 = 4$. This time Nicolas obtains $16$ votes and Ségolène
only $4$. This example clearly  shows that this bug in the condition in
the original paper can have important consequences.

\subsection{Private Multi-Party Trust Computation}

In~\cite{Trust10} the authors give a multiple private keys protocol
for private multi-party computation of a trust value: an initiating
user wants to know the (possibly weighted) average trust the network of
nodes has in some user. In a first phase of the protocol, each of the
$n$ nodes splits its trust $t$ in $n-1$ shares $(s_i)$ such that
\[ t = s_1 + s_2 + \ldots + s_{n-1} \bmod r. \]
Here $r$ is a common modulus chosen big enough with respect to the
maximum possible global trust value, and in order to insure the
privacy of its trust value the shares should be taken as random number
$\bmod~r$, except for the last one. The shares are then sent encrypted
(using Benaloh's scheme) to each other user, to be later
recombined. If we assume that one of the users has chosen a faulty
value for his public parameter $y$, then his contribution to the
recombined value will be computed $\bmod~r'$ instead of $\bmod~r$ for
some divisor $r'$ of $r$. As an extreme example, assume
\begin{itemize}
\item
 that the queried user is a newcomer, untrusted by anyone (hence the
 private value of $t$ for every node is $0$),
\item
  that the true recombined value contributed by the faulty user should
  have been $r-1$,
\item
  that $r'=r/3$.
\end{itemize}
Due to his miscalculation, the faulty node will contribute the value
$r'-1$ instead of $-1$, causing the apparent calculated trust value to
be quite high (about $1/3$ of the maximum possible trust value, instead of
$0$). This can have dramatic consequences if the trust value is used
later on to grant access to some resource. These assumptions are not
entirely unlikely: remember that $r=3^k$ is an explicitly suggested
choice of parameter of the cryptosystem (chosen for instance
in~\cite{PIR05}) in which case $\rho$ is close to $1/3$ and faulty
nodes occur with high probability even with moderate-sized
networks. We note also that the description from~\cite{Benaloh94} is
given \emph{in extenso}, with its incorrect condition. One reason for
choosing Benaloh's cryptosystem in this application is because the
cleartext space can be common among several private keys, a feature
unfortunately not achieved \emph{e.g.} by Paillier's
cryptosystem~\cite{Paillier:1999:PKC} but also possible with
Naccache-Stern's~\cite{NaccacheS98}.

\subsection{Secure Cards Dealing}

Another application of this encryption scheme is given in
\cite{Golle:2005}: securely dealing cards in poker (or similar
games). Here again the author gives the complete description of the
original scheme, with a choice of parameter $r=53$ (which is
prime). Because $r$ is prime, this application does not suffer from
the flaw explained here, but this choice of a prime number is done for
reasons purely internal to the cards dealing protocol, namely testing
the equality of dealt cards.

Given two ciphertext $E(m_1)$ and $E(m_2)$, the players need to test
if $m_1=m_2$ without revealing anything more about the cards $m_1$ and
$m_2$. The protocol is as follows:
\begin{enumerate}
\item
  Let $m=m_1-m_2$, each player can compute $E(m) = E(m_1)/E(m_2)$
  because of the homomorphic property of the encryption.
\item
  Each player $P_i$ secretly picks a value $0<\alpha_i<53$, computes
  $E(m)^{\alpha_i}$ and discloses it to everyone.
\item
  Each player can compute $\prod_i E(m)^{\alpha_i} = E(m)^\alpha$ with
    $\alpha=\sum_i \alpha_i$. The players jointly decrypt
    $E(m)^\alpha$ to get the value $m\alpha\bmod r$.
\end{enumerate}
Now because for each player the value of $\alpha$ is unknown and
random, if $m\alpha\neq 0 \bmod r$ then the players learn nothing
about $m$. Otherwise they conclude that the cards are equal.

We claim that this protocol fails to account for two problems:
\begin{itemize}
\item
  there is no guarantee that $\alpha \neq 0 \bmod r$. When this
  happens, two distinct cards will be incorrectly considered equal.
\item
  knowing the value of $E(m)$ and $E(m)^{\alpha_i}$, it is easy to
  recover $\alpha_i$ because of the small search space for
  $\alpha_i$. This means the protocol leaks information when $m_1\neq
  m_2$. The fix here is to multiply by some random encryption of $0$. 
\end{itemize}

\section{Conclusion}
\label{sec:ccl}

We have shown that the original definition of Benaloh's homomorphic
encryption does not give sufficient conditions in the choice of public
key to get an unambiguous encryption scheme. We gave a necessary and
sufficient condition which fixes the scheme. Our discussion on the
probability of choosing an incorrect public key shows that this
probability is non negligible for parameters where decryption is
efficient: for example using the suggested value of the form $r=3^k$,
this probability is already close to $1/3$.  We also explain on some
examples what can be the consequences of the use of the original
Benaloh scheme. In fact, it is surprising this result was not found
before, considering the number of applications built on the
homomorphic property of Benaloh's scheme. This strongly suggests this
scheme was rarely implemented.

\bibliographystyle{alpha}
\bibliography{biblio}

\end{document}